\documentclass[cmbright]{my_mmaauth}

\usepackage{moreverb}

\usepackage[colorlinks,bookmarksopen,bookmarksnumbered,citecolor=red,urlcolor=red]{hyperref}

\usepackage{amsmath,amssymb,url}
\usepackage{mathrsfs,color}
\usepackage{enumerate,graphicx}
\usepackage{epstopdf}
\usepackage[active]{srcltx}
\usepackage{pstricks}
\usepackage{epsfig}
\usepackage{epstopdf}
\usepackage[capbesideposition=outside,capbesidesep=quad]{floatrow}

\newtheorem{theorem}{Theorem}
\newtheorem{proposition}[theorem]{Proposition}
\newtheorem{remark}[theorem]{Remark}
\newtheorem{definition}[theorem]{Definition}

\newenvironment{proof}{\paragraph{Proof.}}{\hfill$\square$}

\begin{document}

\runninghead{F. Nda\"{\i}rou, I. Area, J. J. Nieto, C. J. Silva, D. F. M. Torres}

\title{Mathematical modeling of Zika disease\\ 
in pregnant women and newborns\\ 
with microcephaly in Brazil}

\author{Fa\"{\i}\c{c}al Nda\"{\i}rou\affil{a,b},
Iv\'an Area\affil{b},
Juan J. Nieto\affil{c},
Cristiana J. Silva\affil{d},\\
Delfim F. M. Torres\affil{d}\corrauth}


\address{\affilnum{a}African Institute for Mathematical Sciences, AIMS-Cameroon, 
P.O. Box 608, Limb\'e Crystal Gardens, South West Region, Cameroon\\
\affilnum{b}Departamento de Matem\'atica Aplicada II, E.E. Aeron\'autica e do Espazo, 
Universidade de Vigo, Campus As Lagoas s/n, 32004 Ourense, Spain\\
\affilnum{c}Facultade de Matem\'{a}ticas, Universidade de Santiago de Compostela, 
15782 Santiago de Compostela, Spain\\
\affilnum{d}Center for Research and Development in Mathematics and Applications (CIDMA), 
Department of Mathematics, University of Aveiro, 3810--193 Aveiro, Portugal}

\corraddr{Delfim F. M. Torres, 
Department of Mathematics,
University of Aveiro,
3810-193 Aveiro, Portugal.
Email: delfim@ua.pt}


\begin{abstract}
We propose a new mathematical model for the spread of Zika virus. 
Special attention is paid to the transmission of microcephaly. 
Numerical simulations show the accuracy of the model 
with respect to the Zika outbreak occurred in Brazil.
\end{abstract}

\MOS{34D20; 92D30} 


\maketitle

\vspace{-6pt}


\section{Introduction}

Zika virus infection on humans is mainly caused by the bite of an infected 
\emph{Aedes} mosquito, either \emph{A. aegypti} or \emph{A. albopictus}. 
The infection on human usually causes rash, mild fever, conjunctivitis, 
and muscle pain. These symptoms are quite similar to dengue and chikungunya diseases, 
which can be transmitted by the same mosquitoes. Other modes of transmission 
of Zika disease have been observed, as sexual transmission, though less common
\cite{Moreno:2017}. Such modes of transmission are included in mathematical
models found in the recent literature: see \cite[Section~8]{Driessche:inpress}
for a good state of the art.

The name of the virus comes from the Zika forest in Uganda, 
where the virus was isolated for the first time in 1947. 
Up to very recent times, most of the Zika outbreaks have occurred 
in Africa with some sporadic outbreaks in Southeast Asia and also 
in the Pacific Islands. Since May 2015, Zika virus infections have been confirmed 
in Brazil and, since October 2015, other countries and territories 
of the Americas have reported the presence of the virus: see 
\cite{nature:2017}, where evolutionary trees, constructed using both newly 
sequenced and previously available Zika virus genomes, reveal how 
the recent outbreak arose in Brazil and spread across the Americas.
The subject attracted a lot of attention and is now under strong current investigations.
In \cite{Agusto:2017}, a deterministic model, 
based on a system of ordinary differential equations, 
was proposed for the study of the transmission dynamics of the Zika virus.
The model incorporates mother-to-child transmission as well as the development 
of microcephaly in newly born babies. The analysis shows that the disease-free equilibrium 
of the model is locally and globally asymptotically stable, whenever the associated
reproduction number is less than one, and unstable otherwise. 
A sensitivity analysis was carried out showing that
the mosquito biting, recruitment and death rates, are among the parameters
with the strongest impact on the basic reproduction number. Then,
some control strategies were proposed with the aim to reduce such values \cite{Agusto:2017}.
A two-patch model, where host-mobility is modeled using a Lagrangian approach, 
is used in \cite{Moreno:2017}, in order to understand the role
of host-movement on the transmission dynamics of Zika virus in
an idealized environment. Here we are concerned with the situation in Brazil
and its consequences on brain anomalies, in particular microcephaly, which occur 
in fetuses of infected pregnant woman. This is a crucial question as far as
the main problem related with Zika virus is precisely the number 
of neurological disorders and neonatal malformations \cite{zika-microcephaly}. 

Our study is based on the Zika virus situation reports for Brazil, 
as publicly available at the World Health Organization (WHO) web page. 
Based on a systematic review  of the literature up to 30th May 2016, 
the WHO has concluded that Zika virus infection during pregnancy is a cause of congenital 
brain abnormalities, including microcephaly. Moreover, another important conclusion 
of the WHO is that the Zika virus is a trigger of Guillain-Barr\'e syndrome \cite{WHO}.
Our analysis is focused on the number of confirmed cases of Zika in Brazil. 
For this specific case, an estimate of the population of the country is known, 
as well as the number of newborns. Moreover, from WHO data, it is possible 
to have an estimation of the number of newborn babies with neurological disorder. 
Our mathematical model allows to predict the number 
of cases of newborn babies with neurological disorder.

The manuscript is organized as follows. In Section~\ref{sec:2},
we introduce the model. Then, in Section~\ref{sec:3}, we prove
that the model is biologically well-posed, in the sense that
the solutions belong to a biologically feasible region
(see Theorem~\ref{lemma:posit:invariant}).
In Section~\ref{sec:4}, we give analytical expressions
for the two disease free equilibria of our dynamical system.
We compute the basic reproduction number $R_0$ of the system 
and study the relevant equilibrium point of interaction between 
women and mosquitoes, showing its local asymptotically stability when 
$R_0$ is less than one (see Theorem~\ref{thm:Stab:DFE}).
The sensitivity of the basic reproduction number $R_0$,
with respect to the parameters of the system, is investigated 
in Section~\ref{sec:sens} in terms of the normalized forward sensitivity index.
The possibility of occurrence of an endemic equilibrium is discussed
in Section~\ref{sec:5}. We end with Sections~\ref{sec:6} and \ref{sec:7}
of numerical simulations and conclusions, respectively.


\section{The Zika model}
\label{sec:2}

We consider women as the population under study. The total women population,
given by $N$, is subdivided into four mutually exclusive compartments, 
according to disease status, namely: susceptible pregnant women ($S$);
infected pregnant women ($I$); women who gave birth to babies without 
neurological disorder ($W$); women who gave birth to babies 
with neurological disorder due to microcephaly ($M$).

As for the mosquitoes population, since the Zika virus is transmitted 
by the same virus as Dengue disease, we shall use the same scheme as in 
\cite{MR3557143}. There are four state variables related to the (female) 
mosquitoes, namely: $A_{m}(t)$, which corresponds to the aquatic phase, 
that includes the egg, larva and pupa stages; $S_{m}(t)$, for the mosquitoes 
that might contract the disease (susceptible); $E_{m}(t)$, for the mosquitoes 
that are infected but are not able to transmit the Zika virus to humans (exposed); 
$I_{m}(t)$, for the mosquitoes capable of transmitting the Zika virus to humans (infected).

The following assumptions are considered in our model:
\begin{enumerate}
\item[(A.1)] there is no immigration of infected humans;

\item[(A.2)] the total human populations $N$ is constant;

\item[(A.3)] the coefficient of transmission of Zika virus 
is constant and does not varies with seasons;

\item[(A.4)] after giving birth, pregnant women 
are no more pregnant and they leave the population under study
at a rate $\mu_h$ equal to the rate of humans birth;

\item[(A.5)] death is neglected, as the period of pregnancy 
is much smaller than the mean humans lifespan;

\item[(A.6)] there is no resistant phase for the mosquito, 
due to its short lifetime.
\end{enumerate}
Note that the male mosquitoes are not considered in this study because 
they do not bite humans and consequently they do not influence the 
dynamics of the disease. The differential system that describes 
the model is composed by compartments of pregnant women and women who gave birth: 
\begin{equation}
\begin{cases}
\label{zikamodel1}
\displaystyle{\frac{d S}{dt} 
= \Lambda - (\phi B\beta_{mh} \frac{I_m}{N}  + (1-\phi) \tau_1 + \mu_h) S},\\[2mm]
\displaystyle{\frac{d I}{dt} 
= \phi  B\beta_{mh} \frac{I_m}{N} S - (\tau_2 + \mu_h) I},\\[2mm]
\displaystyle{\frac{d W}{dt} 
= (1-\phi) \tau_1 S + (1-\psi)\tau_2 I - \mu_h W},\\[2mm]
\displaystyle{\frac{d M}{dt} 
= \psi \tau_2 I - \mu_h M},
\end{cases}
\end{equation}
where $N = S + I + W + M$ is the total population (women). 
The parameter $\Lambda$ denotes the new pregnant women per week, 
$\phi$ stands for the fraction of susceptible pregnant women 
that gets infected, $B$ is the average daily biting (per day), 
$\beta_{mh}$ represents the transmission probability from infected mosquitoes 
$I_m$ (per bite), $\tau_{1}$ is the rate at which susceptible pregnant 
women $S$ give birth (in weeks), $\tau_{2}$ is the rate at which infected 
pregnant women $I$ give birth (in weeks), $\mu_{h}$ is the natural 
death rate for pregnant women, $\psi$ denotes the fraction 
of infected pregnant women $I$ that give birth babies 
with neurological disorder due to microcephaly. 
The above system \eqref{zikamodel1} is coupled 
with the dynamics of the mosquitoes 
\cite{RodriguesMonteiroTorresZinober:Dengue}:
\begin{equation}
\begin{cases}
\label{zikamodel2}
\displaystyle{\frac{dA_m}{dt}= \mu_b\left(1- \frac{A_m}{K}\right) 
\left(S_m +  E_m + I_m\right) - \left(\mu_A + \eta_A\right)A_m},\\[2mm]
\displaystyle{\frac{dS_m}{dt}= \eta_A A_m- \Big( B\beta_{hm} \frac{I}{N} + \mu_m \Big) S_m},\\[2mm]
\displaystyle{\frac{dE_m}{dt}=  \Big( B\beta_{hm} \frac{I}{N}\Big) S_m -(\eta_m + \mu_m) E_m},\\[2mm]
\displaystyle{\frac{dI_m}{dt}= \eta_m E_m  - \mu_m I_m},
\end{cases}
\end{equation}
where parameter $\beta_{hm}$ represents the transmission probability 
from infected humans $I_h$ (per bite), $\mu_b$ stands for the number 
of eggs at each deposit per capita (per day), 
$\mu_A$ is the natural mortality rate of larvae (per day),
$\eta_A$ is the maturation rate from larvae to adult (per day), 
$1/\eta_m$ represents the extrinsic incubation period (in days), 
$1/\mu_m$ denotes the average lifespan of adult mosquitoes (in days), 
and $K$ is the maximal capacity of larvae. See Table~\ref{statevar:parameters} 
for the description of the state variables and parameters of the Zika model 
\eqref{zikamodel1}--\eqref{zikamodel2}. 
\begin{table}[!htb]
\floatbox[\capbeside]{table}
{\caption{Variables and parameters of the Zika model \eqref{zikamodel1}--\eqref{zikamodel2}.}
\label{statevar:parameters}}
\centering
\begin{tabular}{|l|l|} \hline
{\scriptsize{Variable/Symbol}} & {\scriptsize{Description}}  \\ \hline
{\scriptsize{$S(t)$}} & {\scriptsize{susceptible pregnant women}} \\
{\scriptsize{$I(t)$}} & {\scriptsize{infected pregnant women}} \\
{\scriptsize{$W(t)$}} & {\scriptsize{women who gave birth to babies without neurological disorder}} \\
{\scriptsize{$M(t)$}} & {\scriptsize{women who gave birth to babies with neurological disorder due to microcephaly}} \\
{\scriptsize{$A_{m}(t)$}} & {\scriptsize{mosquitoes in the aquatic phase}} \\
{\scriptsize{$S_{m}(t)$}} & {\scriptsize{susceptible mosquitoes}} \\
{\scriptsize{$E_{m}(t)$}} & {\scriptsize{exposed mosquitoes}} \\
{\scriptsize{$I_{m}(t)$}} & {\scriptsize{infected mosquitoes}} \\ \hline
{\scriptsize{$\Lambda$}} & {\scriptsize{new pregnant women (per week)}} \\
{\scriptsize{$\phi$}} & {\scriptsize{fraction of $S$ that gets infected}} \\
{\scriptsize{$B$}} & {\scriptsize{average daily biting (per day)}}  \\
{\scriptsize{$\beta_{mh}$}} & {\scriptsize{transmission probability from $I_m$ (per bite)}}  \\
{\scriptsize{$\tau_{1}$}} & {\scriptsize{rate at which $S$ give birth (in weeks)}}   \\
{\scriptsize{$\tau_{2}$}} & {\scriptsize{rate at which $I$ give birth (in weeks)}}   \\
{\scriptsize{$\mu_{h}$}} & {\scriptsize{natural death rate}}  \\
{\scriptsize{$\psi$}} & {\scriptsize{fraction of $I$ that gives birth to babies with neurological disorder}}  \\
{\scriptsize{$\beta_{hm}$}} & {\scriptsize{transmission probability from $I_h$ (per bite)}}   \\
{\scriptsize{$\mu_b$}} & {\scriptsize{number of eggs at each deposit per capita (per day)}}   \\
{\scriptsize{$\mu_A$}} & {\scriptsize{natural mortality rate of larvae (per day)}}  \\
{\scriptsize{$\eta_A$}} & {\scriptsize{maturation rate from larvae to adult (per day)}}  \\
{\scriptsize{$1/\eta_m$}} & {\scriptsize{extrinsic incubation period (in days)}}  \\
{\scriptsize{$1/\mu_m$}} & {\scriptsize{average lifespan of adult mosquitoes (in days)}}  \\
{\scriptsize{$K$}} & {\scriptsize{maximal capacity of larvae}}  \\ \hline
\end{tabular}
\end{table}

In Figure~\ref{fig:flowchart}, we describe the behavior 
of the movement of individuals among these compartments. 
\begin{figure}[h!]
\centering 
\includegraphics[scale=0.8]{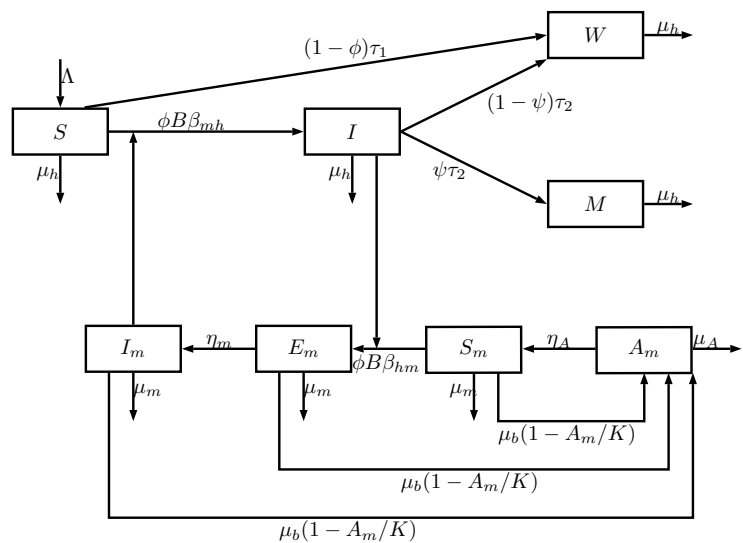}
\caption{Flowchart presentation of the compartmental model 
\eqref{zikamodel1}--\eqref{zikamodel2} for Zika.}
\label{fig:flowchart}
\end{figure}

We consider system \eqref{zikamodel1}--\eqref{zikamodel2} with given initial conditions
\begin{equation*}
S(0) = S_0, \quad I(0) = I_0, \quad W(0) = W_0, \quad M(0) = M_0,
\quad A_m(0) = A_{m0}, \quad S_m(0) = S_{m0}, \quad E_m(0) = E_{m0}, 
\quad I_m(0) = I_{m0},
\end{equation*}
with $\left( S_0, I_0, W_0, M_0, A_{m0}, S_{m0}, E_{m0}, I_{m0} \right) \in \mathbb{R}^{8}_{+}$.
In what follows, we assume $\beta_{mh} = \beta_{hm}$.


\section{Positivity and boundedness of solutions}
\label{sec:3}

Since the systems of equations \eqref{zikamodel1} and \eqref{zikamodel2} 
represent, respectively, human and mosquitoes populations, and all parameters 
in the model are nonnegative, we prove that, given nonnegative 
initial values, the solutions of the system are nonnegative. 
More precisely, let us consider the biologically feasible region
\begin{equation}
\label{eq:Omega}
\Omega = \left\{ (S, I, W, M, A_m, S_m, E_m, I_m) \in \mathbb{R}^8_+ \, : \, S + I + W + M 
\leq \frac{\Lambda}{\mu_h} \, , \quad A_m \leq k N_h \, , \quad S_m + E_m + I_m \leq m N_h \right\}.
\end{equation}
The following result holds.

\begin{theorem}
\label{lemma:posit:invariant}
The region $\Omega$ defined by \eqref{eq:Omega} is positively invariant for model
\eqref{zikamodel1}--\eqref{zikamodel2} with initial conditions
in $\mathbb{R}^{8}_{+}$.
\end{theorem}

\begin{proof}
Our proof is inspired by \cite{RodriguesMonteiroTorresZinober:Dengue}. 
System \eqref{zikamodel1}--\eqref{zikamodel2} 
can be rewritten in the following way:
$\displaystyle \frac{dX}{dt} = M(X) X + F$,
where $X = (S, I, W, M, A_m, S_m, E_m, I_m)$,
$M(X) =
\begin{pmatrix}
M_1 & M_2 
\end{pmatrix}$
with 
\begin{equation*}
M_1 = 
\begin{pmatrix}
-\phi  B\beta_{mh} \frac{I_m}{N}   - (1-\phi) \tau_1 - \mu_h & 0 & 0 & 0 \\
\phi B\beta_{mh} \frac{I_m}{N}  & - \tau_2 - \mu_h &  0 & 0\\
(1-\phi) \tau_1 & - (1-\psi)\tau_2 & - \mu_h & 0\\
0 & \psi \tau_2 & 0 & - \mu_h\\
0 & 0 & 0 & 0 \\
0 & 0 & 0 & 0 \\
0 & 0 & 0 & 0 \\
0 & 0 & 0 & 0
\end{pmatrix} ,
\end{equation*}
\begin{equation*}
M_2 = 
\begin{pmatrix}
0 & 0 & 0 & 0\\
0 & 0 & 0 & 0\\
0 & 0 & 0 & 0\\
0 & 0 & 0 & 0\\
-\mu_b \frac{S_m +  E_m + I_m}{K} - \mu_A - \eta_A & \mu_b & \mu_b & \mu_b\\
\eta_A & - B\beta_{hm} \frac{I}{N} - \mu_m & 0 & 0\\
0 & B\beta_{hm} \frac{I}{N}  & -\eta_m - \mu_m & 0\\
0 & 0 & \eta_m & - \mu_m
\end{pmatrix},
\end{equation*}
and $F = \left( \Lambda, 0, 0, 0, 0, 0, 0, 0 \right)^T$. 
Matrix $M(X)$ is Metzler, i.e., the off diagonal elements of $A$ are nonnegative. 
Using the fact that $F \geq 0$, system $\displaystyle \frac{dX}{dt} = M(X) X + F$ 
is positively invariant in $R^8_+$ \cite{Abate}, which means that any trajectory 
with initial conditions in $R^8_+$ remains in $\Omega$ for all $t > 0$. 
\end{proof}


\section{Existence and local stability of the disease-free equilibria}
\label{sec:4}

System \eqref{zikamodel1}--\eqref{zikamodel2} admits two disease free equilibrium 
points (DFE), obtained by setting the right-hand sides of the equations 
in the model to zero: the DFE $E_1$, given by
\begin{equation*}
E_1 = (S^*, I^*, W^*, M^*, A_m^*, S_m^*, E_m^*, I_m^*)
= \left({\frac {\Lambda}{\tau_1(1- \phi)+{\mu_h}}}, 0, 
\frac {\tau_1 \, \Lambda\, \left( 1-\phi \right) }{\mu_h \, 
\left( \tau_1(1 -\phi)+ \mu_h \right)}, 0, 0, 0, 0, 0 \right),
\end{equation*}
which corresponds to the DFE in the absence of mosquitoes,
and the DFE in the presence of mosquitoes, $E_2$, given by 
\begin{equation*}
E_2 = (S^*, I^*, W^*, M^*, A_m^*, S_m^*, E_m^*, I_m^*) 
= \left( \frac {\Lambda}{\tau_1(1 - \phi) + \mu_h}, 0, \frac{\tau_1 \,
\Lambda\, \left( 1 - \phi \right) }{\mu_h\, \left( \tau_1(1- \phi) + \mu_h \right) }, 
0, -\frac{K \varrho }{\mu_b \eta_A }, -\frac{K \varrho }{\mu_b \mu_{m}}, 0, 0 \right),
\end{equation*}
where
\begin{equation}
\label{eq:varrho}
\varrho = \eta_A (\mu_m - \mu_b) + \mu_A \mu_m.
\end{equation} 
In what follows, we consider only the DFE $E_2$, because this equilibrium point 
considers interaction between humans and mosquitoes, being therefore 
more interesting from the biological point of view. 

The local stability of $E_2$ can be established using 
the next-generation operator method 
on \eqref{zikamodel1}--\eqref{zikamodel2}.
Following the approach of \cite{van:den:Driessche:2002}, 
we compute the basic reproduction number $R_0$ of system 
\eqref{zikamodel1}--\eqref{zikamodel2} writing the right-hand 
side of \eqref{zikamodel1}--\eqref{zikamodel2}
as $\mathcal{F} - \mathcal{V}$ with
\begin{equation*}
\mathcal{F} =
\begin{pmatrix}
0 \\
\phi  B\beta_{mh} \frac{I_m}{N}  S\\
0 \\
0 \\
0 \\
0 \\
 B\beta_{hm} \frac{I}{N}S_m \\
0
\end{pmatrix}, 
\qquad
\mathcal{V} =
\begin{pmatrix}
-\Lambda +(\phi   B\beta_{mh} \frac{I_m}{N}  + (1-\phi) \tau_1 + \mu_h) S\\
( \tau_2 + \mu_h) I\\
-(1-\phi) \tau_1 S + (1-\psi)\tau_2 I + \mu_h W\\
\mu_h M-\psi \tau_2 I \\
-\mu_b(1- \frac{A_m}{K}) (S_m +  E_m + I_m) + (\mu_A + \eta_A)A_m\\
-\eta_A A_m +  (B\beta_{hm} \frac{I}{N}  + \mu_m) S_m\\
(\eta_m + \mu_m) E_m\\
-\eta_m E_m  + \mu_m I_m 
\end{pmatrix}.
\end{equation*}
Then we consider the Jacobian matrices associated with $\mathcal{F}$
and $\mathcal{V}$:
\begin{equation*}
J_{\mathcal{F}} =
\begin{pmatrix} 0&0&0&0&0&0&0&0\\ 
\noalign{\medskip}{
\frac {\phi\,\epsilon \, I_m \, \left( I+W+M \right) }{N}}
&-{\frac {\phi\,\epsilon \, I_m \,S}{ N}}&-{\frac {\phi\,\epsilon I_m \,S}{ N}}
&-{\frac {\phi\,\epsilon I_m \,S}{ N}}&0&0&0&{\phi\,\epsilon\,S}\\ 
\noalign{\medskip}0&0&0&0&0&0&0&0\\ \noalign{\medskip}0&0&0&0&0&0&0&0\\ 
\noalign{\medskip}0&0&0&0&0&0&0&0\\ \noalign{\medskip}0&0&0&0&0&0&0&0\\ 
\noalign{\medskip}-{\frac {\epsilon\, I \, S_m}{N}}&{\frac {\epsilon\, S_m (S + W + M)}{ N}}
&-{\frac {\epsilon I \, S_m}{ N}}&-{\frac {\epsilon\, I\, S_m}{ N}}&0&{\epsilon\, I}&0&{\epsilon\, S_m}\\
\noalign{\medskip}0&0&0&0&0&0&0&0
\end{pmatrix},
\end{equation*}
where $\epsilon=B \beta_{mh}/N$, and
$J_{\mathcal{V}} =
\begin{pmatrix} 
J_{\mathcal{V}1} & J_{\mathcal{V}2}
\end{pmatrix}$
with
\begin{equation*}
J_{\mathcal{V}1} =
\begin{pmatrix} 
A &-{\frac {\phi \epsilon I_m S}{ N}}&-{\frac {\phi\,\epsilon I_m S}{ N}}
&-{\frac {\phi\,\epsilon I_m S}{N}}\\ 
\noalign{\medskip} 0&\tau_2 + \mu_h &0&0\\ 
\noalign{\medskip}\tau_1 \, \left( \phi-1 \right) & \left( \psi-1 \right) \tau_2 & \mu_h &0\\ 
\noalign{\medskip}0&-\psi \tau_2 &0& \mu_h\\ \noalign{\medskip}0&0&0&0\\ 
\noalign{\medskip}-{\frac {\epsilon I  S_m}{ N}}&{\frac {\epsilon S_m (S + W + M) }{ N}}
&-{\frac {\epsilon I S_m}{ N}}&-{\frac {\epsilon I S_m}{ N}}\\ \noalign{\medskip}0&0&0&0\\ 
\noalign{\medskip} 0&0&0&0
\end{pmatrix}, \quad
J_{\mathcal{V}2} =
\begin{pmatrix} 
0&0&0&{\phi\,\epsilon\,S}\\ 
\noalign{\medskip} 0&0&0&0\\ 
\noalign{\medskip} 0&0&0&0\\ 
\noalign{\medskip} 0&0&0&0\\ 
\noalign{\medskip} {\frac { \mu_b  (S_m + E_m + I_m) + (\mu_A+ \eta_A) K}{K}}
&{\frac {\mu_b \, \left( A_m-K \right) }{K}}&{\frac { \mu_b \,
\left( A_m-K \right) }{K}}&{\frac {\mu_b \, \left( A_m-K \right) }{K}}\\ 
\noalign{\medskip}-\eta_A &{\frac {\epsilon I}{N}+ \mu_{m} }&0&0\\ 
\noalign{\medskip} 0&0& \eta_M + \mu_{m} &0\\ 
\noalign{\medskip}0&0&-\eta_M & \mu_{m} 
\end{pmatrix}
\end{equation*}
and
$A = \displaystyle \tau_1(1-\phi)+\mu_h + \frac{\phi B \beta_{mh} I_m (N- S)}{N^2}$.
The basic reproduction number $R_0$ is obtained as the spectral radius of the matrix
$J_{\mathcal{F}} \times (J_{\mathcal{V}})^{-1}$ at the disease-free
equilibrium $E_2$, and is given by
\begin{equation}
\label{R0initial}
R_0 = \frac{\sqrt {-\mu_b \Lambda \left( \mu_h + \tau_2 \right) 
\left( \tau_1(1- \phi)+ \mu_h \right)  \left( \eta_m + \mu_m \right) K\phi \eta_m \, 
\left( \eta_A (\mu_m  - \mu_b) + \mu_A  \mu_m \right) } \beta_{mh}
\mu_h B}{\mu_b \Lambda  \left( \mu_h + \tau_2 \right)  \left(\tau_1(1 - \phi)
+ \mu_h \right)  \left( \eta_m + \mu_m \right) \mu_m}
\end{equation}
or
\begin{equation}
\label{eq:R0}
R_0^2 = \frac{\phi \, B^2 \, \beta_{mh}^{2} \, K \, \mu_h^{2} \, 
\eta_m \left( \mu_b \eta_A  - \mu_m (\mu_A + \eta_A)\right) }{\mu_m^{2} 
\mu_b \,\Lambda  \left( \mu_h + \tau_2 \right) \left( \eta_m + \mu_m \right)  
\left( \tau_1(1 - \phi) + \mu_h  \right) }.
\end{equation}
The disease-free equilibrium $E_2$ is locally asymptotically stable 
if all the roots of the characteristic equation of the linearized 
system associated to \eqref{zikamodel1}--\eqref{zikamodel2} 
at the DFE $E_2$ have negative real parts. 
The characteristic equation associated with $E_2$ is given by 
\begin{equation}
\label{eq:car}
p_1(\lambda) p_2 (\lambda) p_3(\lambda)  p_4(\lambda) = 0 
\end{equation}
with
$p_1(\lambda) = \lambda + \mu_h + \tau_1 (1- \phi)$, 
$p_2(\lambda) =  \left( \lambda+ \mu_h \right) ^{2}$, 
$p_3(\lambda) = -\lambda^{2} \mu_m + \left( -\mu_b  \eta_A - \mu_m^{2}
\right) \lambda + \mu_m ( \mu_m (\mu_A + \eta_A) - \mu_b \eta_A)$
and $p_4(\lambda) = a_3 \lambda^3 + a_2 \lambda^2 + a_1 \lambda + a_0$, where 
\begin{equation*}
\begin{split}
a_0 &=-\frac{ \mu_m^{2} \mu_b \,\Lambda  \left( \mu_h + \tau_2 \right) 
\left( \eta_m + \mu_m \right)  \left( \tau_1(1 - \phi) + \mu_h  \right) 
- \phi \, B^2 \, \beta_{mh}^{2} \, K \, \mu_h^{2} \, \eta_m 
\left( \mu_b \eta_A  - \mu_m (\mu_A + \eta_A)\right)}{\Lambda\,\mu_m^{2} 
\mu_b \left( \tau_1(1- \phi ) + \mu_h \right)},\\
a_1 &= -\frac{\mu_m^{2}+ \left( 2\, \tau_2 + 2\,\mu_h + \eta_m \right) 
\mu_m + \eta_m  (\tau_2 + \mu_h) }{\mu_m}, \quad
a_2 = -{\frac {\mu_h + \eta_m + \tau_2 +2\, \mu_m }{ \mu_m}}, \quad
a_3 = -\frac{1}{\mu_m}.
\end{split}
\end{equation*}
By the Routh--Hurwitz criterion, all the roots of the characteristic equation 
\eqref{eq:car} have negative parts whenever $R_0 < 1$. 
We have just proved the following result. 

\begin{theorem}
\label{thm:Stab:DFE}
The disease free equilibrium in the presence of non-infected mosquitoes, 
$E_2$, is locally asymptotically stable if $R_0 < 1$ and unstable if $R_0 > 1$. 
\end{theorem}


\section{Sensitivity of the basic reproduction number}
\label{sec:sens}

The sensitivity of the basic reproduction number $R_0$ is an important issue 
because it determines the model robustness to parameter values. The sensitivity 
of $R_0$ with respect to model parameters is here measured 
by the so called \emph{sensitivity index}. 

\begin{definition}[See \cite{Chitnis,Kong}]
\label{def:sense}
The normalized forward sensitivity index of a variable $\upsilon$ 
that depends differentiability on a parameter $p$ is defined by
$\displaystyle \Upsilon_{p}^{\upsilon} 
:= \frac{\partial \upsilon}{\partial p} \times \frac{p}{|\upsilon|}$.
\end{definition}

\begin{remark}
If $\Upsilon_{p}^{\upsilon} = + 1$, then an increase (decrease) 
of $p$ by $x\%$ increases (decreases) $\upsilon$ by $x\%$; 
if $\Upsilon_{p}^{\upsilon} = - 1$, then an increase (decrease) 
of $p$ by $x \%$ decreases (increases) $\upsilon$ by $x\%$. 
\end{remark}

From \eqref{R0initial} and Definition~\ref{def:sense}, 
it is easy to derive the normalized forward sensitivity index 
of $R_0$ with respect to the average daily biting $B$ 
and to the transmission probability $\beta_{mh}$
from infected mosquitoes $I_m$ per bite.

\begin{proposition}
The normalized forward sensitivity index 
of $R_0$ with respect to $B$ and $\beta_{mh}$ is one:
$\Upsilon_{B}^{R_0} = 1$ and $\Upsilon_{\beta_{mh}}^{R_0} = 1$.
\end{proposition}

\begin{proof}
It is a direct consequence of \eqref{R0initial} and Definition~\ref{def:sense}. 
\end{proof}

The sensitivity index of $R_0$ with respect to $\phi$, $\eta_m$, $\mu_h$, 
$\eta_A$, $\mu_m$, $\tau_1$, $\tau_2$, $\mu_b$, $\mu_A$,  $K$ 
and $\Lambda$ is given, respectively, by
\begin{gather*}
\Upsilon_{\phi}^{R_0} = {\frac {\mu_h +\tau_1}{2 (\tau_1(1 -\phi) + \mu_h) }}, 
\quad \Upsilon_{\eta_m}^{R_0} = \frac{\mu_m}{ 2(\eta_m+ \mu_m)}, 
\quad \Upsilon_{\mu_h}^{R_0} = \frac { ( \tau_1 (\mu_h  + 2 \tau_2))(1 - \phi)  
+ \tau_2 \mu_h }{ 2 \left( \mu_h +\tau_2\right)\left( \tau_1( 1 -\phi) + \mu_h \right) },
\quad \Upsilon_{\eta_A}^{R_0} = \frac { \left( \mu_m - \mu_b  \right) \eta_A}{2( \mu_m (\eta_A + \mu_A) 
- \eta_A  \mu_b )},\\ 
\Upsilon_{\mu_m}^{R_0} =\frac{\mu_m\left( (\eta_A + \mu_A)(-\eta_m  - 2 \mu_m) + 3 \eta_A \mu_b\right) 
+ 2 \eta_A \eta_m \mu_b}{2 \left( \eta_M + \mu_M  \right)  
\left( \mu_M (\eta_A + \mu_A) - \eta_A \mu_b  \right) },\quad
\Upsilon_{\tau_1}^{R_0} = {\frac { \tau_1 \left( -1+\phi \right) }{2(\tau_1( 1-\phi)+ \mu_h )}}, 
\quad \Upsilon_{\tau_2}^{R_0} = - \frac {\tau_2}{ 2(\mu_h +\tau_2)}, \\
\Upsilon_{\mu_b}^{R_0} = - \frac{\mu_M  \left( \mu_A + \eta_A \right)}{2
\left( \mu_M( \eta_A  + \mu_A) - \eta_A \mu_b \right)},\quad 
\Upsilon_{\mu_A}^{R_0} = \frac {\mu_A \mu_m}{2( \mu_m(\eta_A  + \mu_A) - \eta_A \mu_b )}, 
\quad \Upsilon_{K}^{R_0} = \frac{1}{2}, 
\quad \Upsilon_{\Lambda}^{R_0} = - \frac{1}{2}.
\end{gather*}
In Section~\ref{sec:6}, we compute the previous sensitivity indexes for data from Brazil. 

To analyze the sensitivity of $R_{0}^{2}$ with respect to all the parameters involved, 
we compute appropriate derivatives:
\begin{equation}
\label{derivativesR02}
\begin{aligned}
\frac{\partial R_{0}^{2}}{\partial \phi}
&=\frac{\mu _h+\tau _1}{\phi \left(\mu _h+(1-\phi) \tau _1\right)} R_{0}^{2}, \quad
\frac{\partial R_{0}^{2}}{\partial B}=\frac{2}{B} R_{0}^{2}, \quad
\frac{\partial R_{0}^{2}}{\partial \beta_{mh}}=\frac{2}{\beta_{mh}} R_{0}^{2}, \quad
\frac{\partial R_{0}^{2}}{\partial \eta_{m}}
=\frac{\mu _m}{\eta _m (\mu _m+\eta _m)} R_{0}^{2}, \\
\frac{\partial R_{0}^{2}}{\partial \mu_{b}}
&=\frac{\mu _m \left(\eta _a+\mu _a\right)}{\mu _b \left(\eta _a \mu _b
-\mu _m \left(\eta _a+\mu _a\right)\right)} R_{0}^{2}, \quad
\frac{\partial R_{0}^{2}}{\partial \eta_{A}}
=\frac{\mu _b-\mu _m}{\eta _a \mu _b-\mu _m \left(\eta _a+\mu _a\right)} R_{0}^{2}, \quad
\frac{\partial R_{0}^{2}}{\partial \mu_{A}}
=\frac{\mu _m}{\mu _m \left(\eta _a+\mu _a\right)-\eta _a \mu _b} R_{0}^{2},\\
\frac{\partial R_{0}^{2}}{\partial \mu_{m}}
&=\left( \frac{\eta _a+\mu _a}{\mu _m \left(\eta _a+\mu _a\right)-\eta _a \mu _b}
-\frac{1}{\eta _m+\mu _m}-\frac{2}{\mu _m} \right) R_{0}^{2},\quad 
\frac{\partial R_{0}^{2}}{\partial \Lambda}=-\frac{1}{\Lambda}R_{0}^{2}, \quad
\frac{\partial R_{0}^{2}}{\partial \tau_{2}}=-\frac{1}{\mu_{h}+\tau_{2}} R_{0}^{2}.
\end{aligned}
\end{equation}
In Section~\ref{sec:6} we compute the 
values of these expressions according with the numerical values 
given in Table~\ref{parameters}, up to the $R_{0}^{2}$ factor, 
which appears in all the right-hand expressions 
of \eqref{derivativesR02}, in order to determine 
which are the most and less sensitive parameters. 


\section{Existence and stability analysis of the endemic equilibrium point} 
\label{sec:5}

The system \eqref{zikamodel1}--\eqref{zikamodel2} has one endemic equilibrium (EE)
with biologic meaning whenever $R_0 > 1$. This EE is given by
$E^+ = \left( S^*_+, I^*_+, W^*_+, M^*_+, A^*_{m+}, S^*_{m+}, E^*_{m+}, I^*_{m+} \right)$ 
with $\displaystyle S^*_+ = \frac{\zeta_{S}}{d^*}$,
$\displaystyle I^*_+ = \frac{\zeta_{I}}{ d^* (\mu_h+\tau_2)}$, 
$\displaystyle W^*_+ = \frac{\zeta_{W}}{d^* \, \mu_h \, (\mu_h + \tau_2)}$, 
$\displaystyle M^*_+ = \frac{\zeta_{M}}{d^*  \, \mu_h \, (\mu_h + \tau_2)}$,
$\displaystyle A^*_{m+} = \frac{-K \varrho }{ \mu_b  \eta_A }$, 
$\displaystyle S^*_{m+} = \frac{\zeta_{S_{m}}}{ d_{m}^{*}}$,
$\displaystyle E^*_{m+} =\frac{\zeta_{E_{m}} }{ (\eta_{m}+\mu_{m}) d_{m}^{*}}$,
$\displaystyle I^*_{m+} = E^*_{m+} \frac{\eta_{m}}{\mu_{m}}$,
where $\varrho$ is defined in \eqref{eq:varrho} and
\begin{align*}
d^*&= \left[ BK \beta_{mh}\, \eta_m \,\phi\,\mu_h \, \varrho 
-\Lambda\,\mu_m \mu_b  \left( \eta_m \,\tau_1\, \left( 1-
\phi \right) + \mu_m \,\tau_1\, \left( 1-\phi \right) + \mu_h \,
\left( \eta_m + \mu_m \right)  \right)  \right] \mu_h\,\beta_{mh}\,B,\\
d_{m}^{*}&=\mu_b \left( B \beta_{mh} \mu_h 
+ \mu_m  \mu_h + \mu_m \tau_2 \right) \eta_m \phi B \beta_{mh} \mu_h > 0,\\
\zeta_{S}&=-\mu_b\, \left( \eta_m + \mu_m \right)  \left( B \beta_{mh}\,\mu_h 
+ \mu_m  (\mu_h + \tau_2) \right) \mu_m 	\Lambda^{2}, \\
\zeta_{I}&=\Lambda\, \left( {B}^{2}K \beta_{mh}^{2} \eta_m\,\phi\, \mu_h^{2}  \varrho 
+\Lambda\, \mu_m^{2} \mu_b \, \left( \mu_h + \tau_2 \right)  
\left(  \left( 1-\phi \right) \tau_1 + \mu_h \right)  \left( \eta_m + \mu_m \right)  \right) , \\
\zeta_{W}&= \Lambda\,K \beta_{mh}^{2} \eta_m  \mu_h^{2} \phi \tau_2 
\left( 1-\psi \right)  \varrho  B^2 - \Lambda^2 \beta_{mh} \mu_m \mu_b \mu_h \tau_1  
\left( 1-\phi \right)  \left( \mu_h +\tau_2 \right)  \left( \eta_m + \mu_m \right) B \\ 
&\quad - (\Lambda^2 \mu_m^2 \mu_b \left(\mu_h + \tau_2 \right)  
\left( \mu_h \tau_2 \left( -1 + \psi \right) +\psi \tau_1 \tau_2 
\left( 1-\phi \right) +\mu_h \tau_1 \left( 1-\phi \right)\right)  
\left( \eta_m + \mu_m  \right)), \\
\zeta_{M}&=\Lambda\, \left( \mu_h^{2} \varrho  
K \beta_{mh}^{2}{B}^{2} \eta_m \phi+\Lambda \mu_m^{2} \mu_b 
\left( \mu_h +\tau_2 \right)  \left( \tau_1\, \left( 1-\phi \right) 
+ \mu_h \right)  \left( \eta_m + \mu_m \right)  \right) \psi\,\tau_2, \\
\zeta_{S_{m}}&=  -\left( \mu_h +\tau_2 \right)  \left( BK \beta_{mh} \eta_m
\phi \mu_h \varrho  -\Lambda \mu_m \mu_b \left( (1-\phi)\,\tau_1+ \mu_h \right)  
\left( \eta_m + \mu_m \right)  \right), \\
\zeta_{E_{m}}&=-\mu_h^{2} \varrho
K \beta_{mh}^{2} {B}^{2} \eta_m \phi - \left( \mu_h + \tau_2 \right) \Lambda \mu_m^{2} 
\mu_b \left( \eta_m + \mu_m \right)  \left( (1-\phi)\,\tau_1 + \mu_h \right).
\end{align*}
From \eqref{eq:R0}, $\mu_b \eta_A > \mu_m (\mu_A + \eta_A)$, $\phi < 1$, and $d^* < 0$. 
Thus, $S^*_+ >0$, $I^*_+ = \displaystyle \frac{\left( 1-R_0^2 \right)\Lambda^{2} \mu_{b} \mu_{m}^{2} 
\left((1-\phi)\tau_{1}+\mu_{h}\right)(\eta_{m}+\mu_{m})}{ d^*}>0$, and
$M^*_+ = \displaystyle \frac{\left( 1-R_0^2 \right) \Lambda^{2}  \mu_b {\mu_m^{2}} \,
\psi\tau_2 \, \left(  \left( 1-\phi \right) \tau_1
+ \mu_h \right)  \left( \eta_m + \mu_m \right)  }{ d^{*} \, \mu_{h}}
= \frac{\psi  \tau_{2}}{\mu_{h}} I_{+}^{*}>0$. 
Moreover, as for $W_{+}^{*}$, we have that it can be expressed as
$\displaystyle W^*_+ = \frac{\varpi}{d^* \, \mu_h \, (\mu_h + \tau_2)}$ with
\begin{multline*}
\varpi=B^2 \beta_{mh}^2 \eta_{m} K \Lambda \mu_{h}^2 
\phi(\psi -1) \tau_{2} (\eta_{A} \mu_{b}-\mu_{m}
(\eta_{A}+\mu_{A})) 
+ \Lambda^2 \mu_{b} \mu_{m} (\eta_{m}+\mu_{m}) (\mu_{h}+\tau_{2})\\
\times \left[ (\phi -1) (\tau_{1}-\tau_{2}) 
(B \beta_{mh} \mu_{h}+\mu_{m} (\mu_{h}+\psi  \tau_{2}))  
+ \tau_{2} (B \beta_{mh} \mu_{h} (\phi -1)
+\mu_{h} \mu_{m} (\phi -\psi )+\mu_{m} (\phi-1)\psi\tau_{2})\right].
\end{multline*}
Therefore, $W^*_+$ is positive assuming that $\tau_1 >\tau_2 $ and $\psi > \phi$.
Finally, by using again $\mu_b \eta_A > \mu_m (\mu_A + \eta_A)$ and $1-\phi > 0$, 
we obtain that $A^*_{m+} > 0$ and, moreover,
\begin{equation*}
S^*_{m+} = \frac{-(\mu_h + \tau_2)d^*}{B \beta_{mh} \mu_{h} d_{m}^{*}} > 0, 
\quad E^*_{m+} = \frac{\Lambda\, \left( R_0^2 -1 \right)  \mu_m^{2} \mu_b \, 
\left( \mu_h + \tau_2 \right)  \left(  \left( 1-\phi \right) \tau_1
+ \mu_h \right)   }{d_{m}^{*}} >0,
\quad I^*_{m+} =E^*_{m+} \frac{\eta_{m}}{\mu_{m}}> 0.
\end{equation*}
After some appropriate manipulations, 
the matrix associated to $E_{+}$ is given by 
\begin{equation}
\begin{pmatrix}
V_{11} & V_{12} & V_{12} & V_{12} & 0 & 0 & 0 & V_{18} \\
V_{21} & -\mu_{h}-\tau_{2}-V_{12} & -V_{12} & -V_{12} & 0 & 0 & 0 & -V_{18} \\
\tau_{1}(1-\phi  )& \tau_{2}(1-\psi  ) & -\mu_{h} & 0 & 0 & 0 & 0 & 0 \\
0 & \psi  \tau_{2} & 0 & -\mu_{h} & 0 & 0 & 0 & 0 \\
0 & 0 & 0 & 0 & V_{45} & V_{56} & V_{56} & V_{56} \\
V_{61} & V_{62} & V_{61} & V_{61} & \eta_{a} & -\mu_{m}-V_{76} & 0 & 0 \\
-V_{61} & -V_{62} & -V_{61} & -V_{61} & 0 & V_{76} & -\eta_{m}-\mu_{m} & 0 \\
0 & 0 & 0 & 0 & 0 & 0 & \eta_{m} & -\mu_{m} \\
\end{pmatrix}
\end{equation}
with
\begin{align*}
V_{11}&=-\mu_{h}+(\phi-1) \tau_{1}- \frac{\mu_{h} \zeta_{I} + \zeta_{S}
+\zeta_{W} }{\mu_{h} \zeta_{S})(\mu_{h}+\tau_{2})} V_{12}, \quad
V_{21}=\frac{B \beta_{mh} d^{*} \eta_{m} \mu_{h} \zeta_{I_{m}} 
\left(\mu_{h} \zeta_{I}+\zeta_{S}+\zeta_{W}\right) 
\phi (\mu_{h}+\tau_{2})}{d_{m}^{*} \mu_{m}(\eta_{m}+\mu_{m})(\zeta_{M}
+\zeta_{W} +\mu_{h}(\zeta_{I}+\zeta_{S} (\mu_{h}+\tau_{2})))^{2}}, \\
V_{45}&=-\eta_{a}-\mu_{a}-\frac{\mu_{b}(\mu_{m} \zeta_{E_{m}} 
+ \eta_{m} \zeta_{I_{m}}+\mu_{m}(\eta_{m}+\mu_{m} 
\zeta_{S_{m}})}{d_{m}^{*} K \mu_{m}(\eta_{m}+\mu_{m})}.
\end{align*}
With these notations, we have that the eigenvalues of the matrix are
\begin{equation*}
\lambda_{1}=\mu_{h}, \quad
\lambda_{2}=\frac{1}{2} \left(V_{45}-\mu_{m}
+\sqrt{4 \eta_{a} V_{56}+(\mu_{m}+V_{45})^2}\right),\quad
\lambda_{3}=\frac{1}{2} \left(V_{45}-\mu_{m}
-\sqrt{4 \eta_{a} V_{56}+(\mu_{m}+V_{45})^2}\right)
\end{equation*}
and the roots of a polynomial of degree five. 
Equivalently, we can write $\lambda_{2}$ 
and $\lambda_{3}$ as follows:
\begin{equation*}
\lambda_{2}=\frac{\sqrt{\eta_{a}^2 \mu_{b}^2+4 \mu_{m}^3 
(\eta_{a}+\mu_{a})-2 \eta_{a} \mu_{b} \mu_{m}^2+\mu_{m}^4}
-\eta_{a}\mu_{b}-\mu_{m}^2}{2 \mu_{m}}, \ 
\lambda_{3}=-\frac{\sqrt{\eta_{a}^2 \mu_{b}^2+4 \mu_{m}^3 
(\eta_{a}+\mu_{a})-2 \eta_{a} \mu_{b} \mu_{m}^2+\mu_{m}^4}
+\eta_{a}\mu_{b}+\mu_{m}^2}{2 \mu_{m}}.
\end{equation*}
Moreover, the polynomial of degree five has leading coefficient one, being given
by $x^{5}+\kappa_{4} x^{4}+\kappa_{3} x^{3}+\kappa_{2} x^{2} 
+ \kappa_{1} x + \kappa_{0}$ with
\begin{align*}
\kappa_{4}&= \eta _m+2 \mu _h+2 \mu _m+\tau _2-V_{11}+V_{12}+V_{76}, \\
\kappa_{3} & =  \eta _m \left(2 \mu _h+\mu _m+\tau _2-V_{11}+V_{12}
+V_{76}\right)+\mu _h \left(4 \mu _m+\tau _2-2 V_{11}+V_{12}+2 V_{76}\right)+\mu _h^2 \\
& +V_{76} \left(\mu _m+\tau_2-V_{11}+V_{12}\right)+2 \tau _2 \mu _m
-2 V_{11} \mu _m+2 V_{12} \mu _m+\mu _m^2+\phi  \tau _1 V_{12}-\tau _2 V_{11} \\
&-\tau _1 V_{12}+\tau _2 V_{12}-V_{12} V_{21}-V_{11} V_{12}, \\
\kappa_{2} &=  \eta _m \left(\mu _h \left(2 \mu _m+\tau _2-2 V_{11}
+V_{12}+2 V_{76}\right)+\mu _h^2+\mu _m \left(\tau _2-V_{11}+V_{12}\right)
+(\phi -1) \tau _1 V_{12}-\tau_2 V_{11}+\tau _2 V_{12} \right. \\
& \left. +V_{76} \left(\tau _2-V_{11}+V_{12}\right)-V_{11} V_{12}
-V_{12} V_{21}+V_{18} (V_{61}-V_{62})\right)+\mu _h 
\left(2 \mu _m \left(\tau2-2 V_{11}+V_{12}+V_{76}\right)+2 \mu _m^2  \right. \\
& \left. +(\phi -1) \tau _1 V_{12}-\tau _2 V_{11}
+V_{76} \left(\tau _2-2 V_{11}+V_{12}\right)-V_{11} V_{12}-V_{12} V_{21}\right)
+\mu_h^2 \left(2 \mu _m-V_{11}+V_{76}\right)+2 \phi  \tau _1 V_{12} \mu _m \\
& +V_{76} \left(\mu _m \left(\tau _2-V_{11}+V_{12}\right)+(\phi -1) \tau _1 V_{12}
+\tau _2\left(V_{12}-V_{11}\right)-\left(V_{11}+V_{21}\right) V_{12}\right)
-2 \tau _2 V_{11} \mu _m-2 \tau _1 V_{12} \mu _m \\
&+2 \tau _2 V_{12} \mu _m+\tau _2 \mu _m^2-V_{11}\mu _m^2
+V_{12} \mu _m^2-2 V_{11} V_{12} \mu _m-2 V_{12} V_{21} \mu_m
+\phi  \tau _1 \tau _2 V_{12}-\tau _1 \tau _2 V_{12} \\
&-\tau _2 V_{11} V_{12}-\tau _2 V_{12} V_{21}, \\
\kappa_{1}&= \eta _m \left(\mu _h \left(\mu _m \left(\tau_2
-2 V_{11}+V_{12}\right)+(\phi -1) \tau _1 V_{12}-\tau _2 V_{11}
+V_{76} \left(\tau _2-2 V_{11}+V_{12}\right)-V_{12} V_{11}
-V_{12} V_{21} \right. \right. \\
& \left. \left. +V_{18} \left(2 V_{61}-V_{62}\right)\right)
+\mu _h^2 \left(\mu _m-V_{11}+V_{76}\right)-\mu _m 
\left(V_{12} \left(-\phi  \tau _1+\tau_1+V_{11}+V_{21}\right)
+\tau _2 \left(V_{11}-V_{12}\right)\right) \right. \\
& \left. +\phi  \tau _1 \tau _2 V_{12}-\phi  \tau _1 V_{18} V_{61}
-V_{76} \left(V_{12} \left(-\phi  \tau_1+\tau_1+V_{11}+V_{21}\right)
+\tau _2 \left(V_{11}-V_{12}\right)\right)-\tau _1 \tau _2 V_{12}
-\tau _2 V_{11} V_{12} \right. \\
& \left.  -\tau _2 V_{12} V_{21}+\tau _1 V_{18} V_{61}+V_{11} V_{18} V_{62}
+V_{18} V_{21} V_{62}\right)+\mu _m \left(\mu _h \left(\mu _m \left(\tau_2
-2 V_{11}+V_{12}\right) \right. \right. \\
&\left. \left. -2 \left(V_{12} \left(-\phi  \tau _1+\tau_1+V_{11}+V_{21}\right)
+\tau _2 V_{11}\right)\right)+\mu _h^2 \left(\mu _m-2 V_{11}\right)
-\mu _m \left(V_{12} \left(-\phi  \tau _1
+\tau _1+V_{11}+V_{21}\right) \right. \right. \\
& \left. \left. +\tau_2 \left(V_{11}-V_{12}\right)\right)
+2 \tau_2 V_{12} \left((\phi -1) \tau_1-V_{11}-V_{21}\right)\right)
+V_{76} \left(\mu _h \left(\mu _m \left(\tau _2-2
V_{11}+V_{12}\right)+(\phi -1) \tau _1 V_{12} \right. \right. \\
& \left. \left. -V_{11} \left(\tau _2+V_{12}\right)-V_{21} V_{12}\right)
+\mu _h^2 \left(\mu _m-V_{11}\right)-\mu _m \left(V_{12}
\left(-\phi  \tau _1+\tau _1+V_{11}+V_{21}\right) \right. \right. \\
& \left. \left. +\tau _2 \left(V_{11}-V_{12}\right)\right)+\tau _2 V_{12} 
\left((\phi -1) \tau _1-V_{11}-V_{21}\right)\right), \\
\kappa_{0}&= \eta _m \left(-\mu _h \left(V_{12} \mu _m \left(-\phi  \tau _1
+\tau _1+V_{11}+V_{21}\right)+\tau _2 V_{11} \mu _m
+(\phi -1) \tau _1 V_{18} V_{61} \right. \right. \\
&\left. \left. +V_{12} V_{76}   \left(-\phi  \tau _1+\tau_1
+V_{11}+V_{21}\right)+\tau _2 V_{76} V_{11}-\tau _2 V_{18} V_{61}
-V_{18} \left(V_{11}+V_{21}\right) V_{62}\right)  \right. \\
&\left. +\mu _h^2 \left(-\left(V_{11} \left(\mu _m+V_{76}\right)
-V_{18} V_{61}\right)\right)+\tau _2 \left(V_{12} \left(\mu _m+V_{76}\right)
-V_{18} V_{61}\right) \left((\phi -1) \tau_1-V_{11}-V_{21}\right)\right) \\
&-\mu _m \left(\mu _h+\tau _2\right) \left(\mu _m+V_{76}\right) 
\left(V_{11} \mu _h+V_{12} \left((1-\phi) +V_{11}+V_{21}\right)\right).
\end{align*}
Obviously, these expressions become rather long. As a consequence, it is not possible 
to use the Routh--Hurwitz criterion in this general setting, but only 
for particular values of the parameters. Moreover, the eigenvalues 
can be computed numerically for the specific values given in 
Table~\ref{parameters} of the next section (realistic data from Brazil), 
and they are given by $\lambda_{1}=-0.02$, $\lambda_{2}=-0.02$, 
$\lambda_{3}=-5000$, $\lambda_{4}=-22.3938$, 
$\lambda_{5}=-0.0511697$, $\lambda_{6}=-0.044689$, 
$\lambda_{7}=-0.00919002$, $\lambda_{8}=-0.0079988$.
We can observe that local stability of the endemic equilibrium holds, 
since all eigenvalues are negative real numbers.


\section{Numerical simulations: case study of Brazil}
\label{sec:6}

We perform numerical simulations to compare the results 
of our model with real data obtained from several reports 
published by the World Health Organization (WHO) \cite{who:data:zika}, 
from the starting point when the first cases of Zika 
have been detected in Brazil and for a period of 40 weeks 
(from February 4, 2016 to November 10, 2016), 
which is assumed to be a regular pregnancy time.


\subsection{Zika model fits well real data}

According to several sources, the total population of Brazil is 206,956,000, 
and every year there are about 3,073,000 new born babies. As a consequence, 
there are about 3,000,000/52 new pregnant females every week. The number 
of babies with neurological disorders is taken from WHO reports 
\cite{who:data:zika}. See Table~\ref{parameters}, where the values 
considered in this manuscript have been collected, such that the numerical 
experiments give good approximation of real data obtained 
from the WHO \cite{who:data:zika}.
\begin{table}[!htb]
\floatbox[\capbeside]{table}
{\caption{Parameter values for system \eqref{zikamodel1}--\eqref{zikamodel2}.}\label{parameters}} 
\centering
\begin{tabular}{|l | l | l |  } \hline
{\scriptsize{Symbol}} & {\scriptsize{Description}}  & {\scriptsize{Value}}  \\ \hline
{\scriptsize{$\Lambda$}} & {\scriptsize{new pregnant women (per week)}}  & {\scriptsize{$3000000/52$ }}  \\
{\scriptsize{$\phi$}} & {\scriptsize{fraction of $S$ that gets infected}}  & {\scriptsize{$0.459$}} \\
{\scriptsize{$B$}} & {\scriptsize{average daily biting (per day)}}  & {\scriptsize{$1$}} \\
{\scriptsize{$\beta_{mh}$}} & {\scriptsize{transmission probability from $I_m$ (per bite)}}  & {\scriptsize{$0.6$}} \\
{\scriptsize{$\tau_{1}$}} & {\scriptsize{rate at which $S$ give birth (in weeks)}}  & {\scriptsize{$37$}} \\
{\scriptsize{$\tau_{2}$}} & {\scriptsize{rate at which $I$ give birth (in weeks)}}  & {\scriptsize{$1/25$}}  \\
{\scriptsize{$\mu_{h}$}} & {\scriptsize{natural death rate}}  & {\scriptsize{$1/50$}} \\
{\scriptsize{$\psi$}} & {\scriptsize{fraction of $I$ that gives birth to babies with neurological disorder}}  
& {\scriptsize{$0.133$}}  \\
{\scriptsize{$\beta_{hm}$}} & {\scriptsize{transmission probability from $I_h$ (per bite)}}  & {\scriptsize{$0.6$}} \\
{\scriptsize{$\mu_b$}} & {\scriptsize{number of eggs at each deposit per capita (per day)}}  & {\scriptsize{$80$}} \\
{\scriptsize{$\mu_A$}} & {\scriptsize{natural mortality rate of larvae (per day)}}  & {\scriptsize{$1/4$}} \\
{\scriptsize{$\eta_A$}} & {\scriptsize{maturation rate from larvae to adult (per day)}}  & {\scriptsize{$0.5$}}  \\
{\scriptsize{$1/\eta_m$}} & {\scriptsize{extrinsic incubation period (in days)}}  & {\scriptsize{$125$}} \\
{\scriptsize{$1/\mu_m$}} & {\scriptsize{average lifespan of adult mosquitoes (in days)}}  & {\scriptsize{$125$}} \\
{\scriptsize{$K$}} & {\scriptsize{maximal capacity of larvae}}  & {\scriptsize{1.09034e+06}}  \\ \hline
\end{tabular}
\end{table}

Figure~\ref{f:comparison} shows how our model fits the real data in the period 
from February 4, 2016 to November 10, 2016. More precisely, the $\ell_{2}$ norm 
of the difference between the real data and the curve produced by our model, 
in the full period, is 392.5591, which gives an average of about 9.57 cases 
of difference each week. We have considered as initial values $S_0 = 2,180,686$ 
($S_0$ is the number of newborns corresponding to the simulation period) 
and the number of births in the period, $I_{0}=1$, $M_0=0$, 
and $W_0=0$ for the women populations, 
and $A_{m0}=S_{m0}=I_{m0}=\text{1.0903e+06}$, 
and $E_{m0}=\text{6.5421e+06}$ for the mosquitoes populations. 
The system of differential equations has been solved by using the  
\texttt{ode45} function of \textsf{Matlab}, in a \textsf{MacBook Pro} 
computer with a 2,8 GHz Intel Core i7 processor 
and 16 GB of memory 1600 MHz DDR3.
\begin{figure}[!ht]
\floatbox[{\capbeside\thisfloatsetup{capbesideposition={right,center},capbesidewidth=4cm}}]{figure}[\FBwidth]
{\caption{Number of newborns with microcephaly. 
The red line corresponds to the real data obtained 
from the WHO \cite{who:data:zika} from 04/02/2016 to 10/11/2016 
and the blue line has been obtained by solving numerically 
the system of ordinary differential equations \eqref{zikamodel1}--\eqref{zikamodel2}. 
The $\ell_{2}$ norm of the difference between the real data and our prediction 
is $992.5591$, which gives an error of less than $9.57$ cases per week.}\label{f:comparison}}
{\includegraphics[scale=0.38]{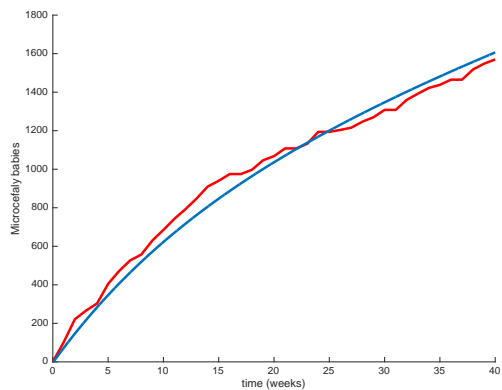}}
\end{figure}


\subsection{Local stability of the endemic equilibrium}

In order to illustrate the local stability of the endemic equilibrium, 
we show that for different initial conditions the solution of the differential 
system \eqref{zikamodel1}--\eqref{zikamodel2} tends to the endemic equilibrium point. 
We have done a number of numerical experiments. Precisely, we consider here four 
different initial conditions for the women population, $C_{1}$, $C_{2}$, $C_{3}$, $C_{4}$, 
defined in Table~\ref{table:eee}, giving the same value for $S_{0}+I_{0}+M_{0}+W_{0}$. 
The system of differential equations \eqref{zikamodel1}--\eqref{zikamodel2} 
has been solved for a period of 156 weeks, which correspond approximately to 3 years, 
by using the same numerical method and machine as for the comparison between 
our model and real data provided by the World Health Organization 
in Figure~\ref{f:comparison}. These numerical experiments 
are included in Figure~\ref{f:equilibrium}.
\begin{table}[htb!]
\floatbox[\capbeside]{table}
{\caption{Initial conditions used in the numerical simulations
for the case study of Brazil (see Figure~\ref{f:equilibrium}).}\label{table:eee}}
{\centering
\begin{tabular}{|c|c|c|c|c|} \hline
& {\scriptsize{$S_{0}$}} & {\scriptsize{$I_{0}$}} 
& {\scriptsize{$M_{0}$}} & {\scriptsize{$W_{0}$}} \\ \hline
{\scriptsize{$C_{1}$}} & {\scriptsize{$2.18069\times 10^6$}} 
& {\scriptsize{1}} & {\scriptsize{0}} & {\scriptsize{0}} \\ \hline
{\scriptsize{$C_{2}$}} & {\scriptsize{$2.17633\times 10^6$}} 
& {\scriptsize{1454.79}} & {\scriptsize{1453.79}} & {\scriptsize{1453.79}} \\ \hline
{\scriptsize{$C_{3}$}} & {\scriptsize{$2.17851\times 10^6$}} 
& {\scriptsize{727.896}} & {\scriptsize{726.896}} & {\scriptsize{726.896}} \\ \hline
{\scriptsize{$C_{4}$}} & {\scriptsize{$2.18025\times 10^6$}} 
& {\scriptsize{146.379}} & {\scriptsize{145.379}} & {\scriptsize{146.379}} \\ \hline
\end{tabular}}
\end{table}
\begin{figure}[h!]
\floatbox[{\capbeside\thisfloatsetup{capbesideposition={right,center},capbesidewidth=4cm}}]{figure}[\FBwidth]
{\caption{The solution of differential system \eqref{zikamodel1}--\eqref{zikamodel2} 
tends to the endemic equilibrium point, independently of initial conditions.
In this figure, we show the evolution of the populations of infected women ($I$) 
and cases of microcephaly ($M$). Initial conditions 
are those of Table~\ref{table:eee}.}\label{f:equilibrium}}
{\centering
\includegraphics[width=0.37\textwidth]{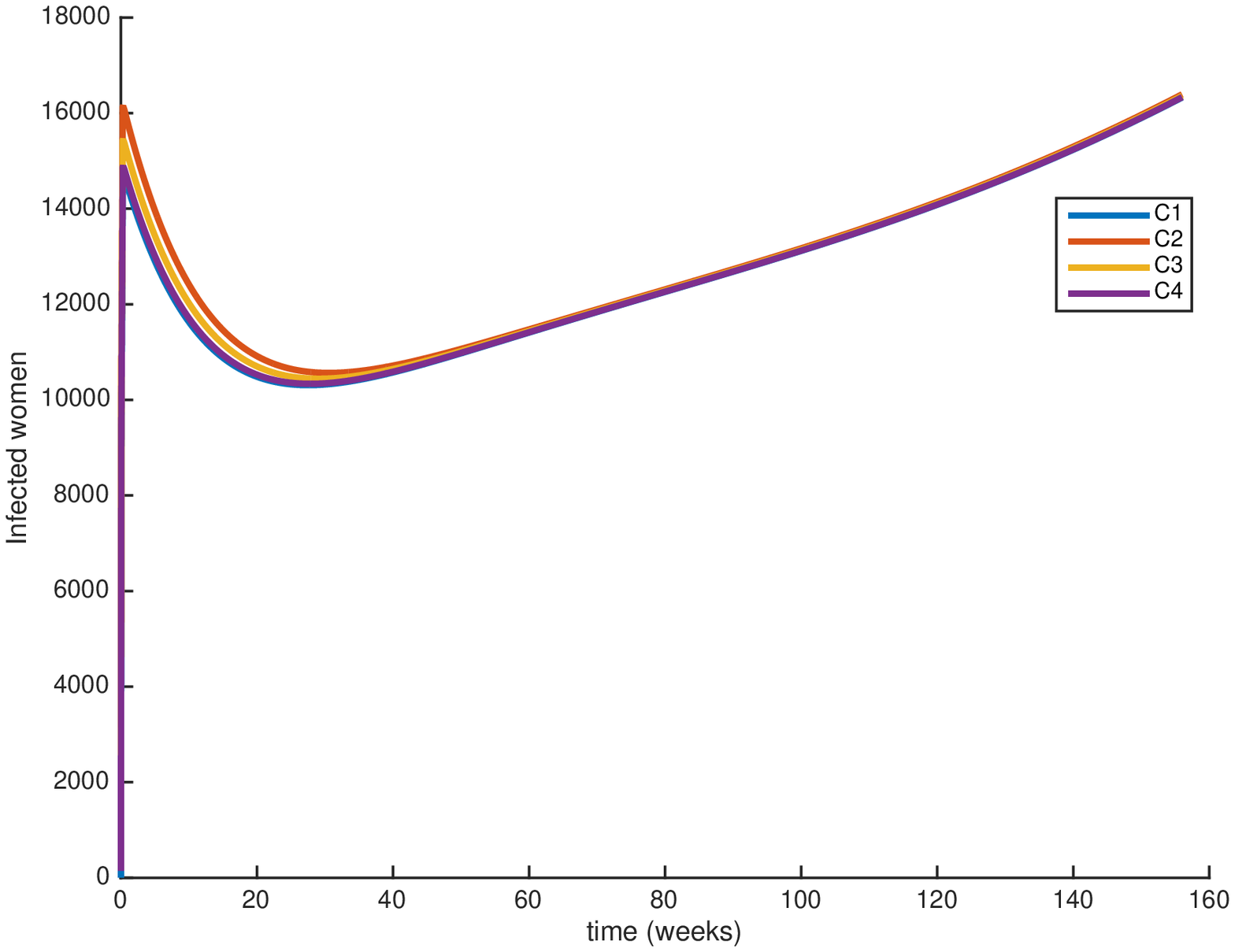}
\includegraphics[width=0.37\textwidth]{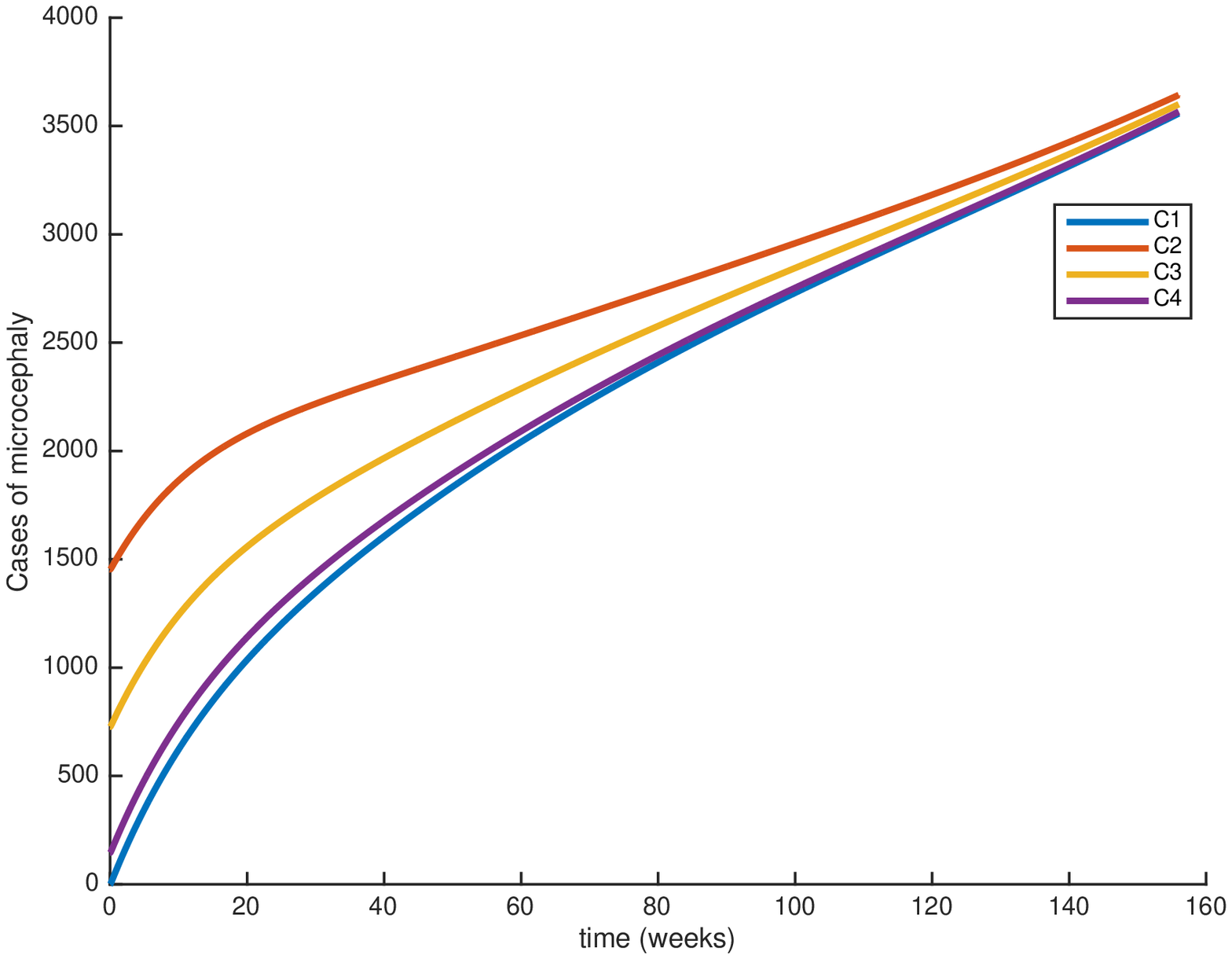}}
\end{figure}


\subsection{Sensitivity analysis of the basic reproduction number}

In order to determine which are the most and less sensitive parameters
of our model, we compute the values of expressions 
\eqref{derivativesR02} according to the numerical values 
given in Table~\ref{parameters}. The numerical results are as follows:
\begin{align*}
\frac{\partial R_{0}^{2}}{\partial \phi}&=4.02523 R_{0}^{2}, \quad
\frac{\partial R_{0}^{2}}{\partial B}=2 R_{0}^{2}, \quad
\frac{\partial R_{0}^{2}}{\partial \beta_{mh}}=3.33333 R_{0}^{2}, \quad
\frac{\partial R_{0}^{2}}{\partial \eta_{m}}=62.5 R_{0}^{2}, \quad
\frac{\partial R_{0}^{2}}{\partial \mu_{b}}=1.8753\,\times10^{-6} R_{0}^{2}, \\
\frac{\partial R_{0}^{2}}{\partial \eta_{a}} &=2.0001 R_{0}^{2}, \quad
\frac{\partial R_{0}^{2}}{\partial \mu_{a}}=-0.00020003 R_{0}^{2},\quad
\frac{\partial R_{0}^{2}}{\partial \mu_{m}} =-312.51 R_{0}^{2}, \quad
\frac{\partial R_{0}^{2}}{\partial \Lambda} =-0.0000173333 R_{0}^{2}, \quad
\frac{\partial R_{0}^{2}}{\partial \tau_{2}} =-16.6667 R_{0}^{2}.
\end{align*}
In Table~\ref{table:sensitivity:index}, we present the sensitivity index of parameters 
$\Lambda$, $\phi$, $B$, $\beta_{mh}$, $\tau_1$, $\tau_2$, $\mu_b$, $\mu_A$, $\eta_A$, 
$\eta_m$, $\mu_m$ and $K$ computed  for the parameter values given in Table~\ref{parameters}. 
\begin{table}[!htb]
\floatbox[\capbeside]{table}
{\caption{The normalized forward sensitivity index 
of the basic reproduction number $R_0$, for the parameter values 
of Table~\ref{parameters}.}\label{table:sensitivity:index}}
{\centering
\begin{tabular}{|l l || l l|} \hline 
{\small{Parameters}} & {\small{Sensitivity index}} 
& {\small{Parameters}} & {\small{Sensitivity index}}\\ \hline
{\scriptsize{$\Lambda$}}, {\scriptsize{$K$}} & {\scriptsize{$-1/2$}} 
& {\scriptsize{$\mu_b$}} & {\scriptsize{$0.00007501125170$ }}   \\
{\scriptsize{$\phi$}} & {\scriptsize{$0.9237909865$ }} & {\scriptsize{$\mu_A$}} 
& {\scriptsize{-0.2500375056e-4 }} \\
{\scriptsize{$B$}}, {\scriptsize{$\beta_{mh}$}} & {\scriptsize{$1$ }} 
& {\scriptsize{$\eta_A$}} & {\scriptsize{$0.5000250040$ }}   \\
{\scriptsize{$\tau_{1}$}} & {\scriptsize{$-0.4995009233$ }} 
& {\scriptsize{$\eta_m$}} & {\scriptsize{$1/4$ }}  \\
{\scriptsize{$\tau_{2}$}} & {\scriptsize{$-1/3$ }} 
& {\scriptsize{$\mu_m$}} & {\scriptsize{$-1.250075011$}} \\ \hline
\end{tabular}}
\end{table}
From Table~\ref{table:sensitivity:index}, we conclude that the most 
sensitive parameters are $B$ and $\beta_{mh}$, which means 
that in order to decrease the basic reproduction number 
in $x\%$ we need to decrease these parameter values in $x\%$. 
Therefore, in order to reduce the transmission of the Zika virus 
it is crucial to implement control measures that lead to a reduction 
on the number of daily biting (per day) $B$ and the transmission 
probability from the infected mosquitoes, $\beta_{mh}$. 
The fraction $\phi$ of susceptible pregnant women $S$ that get 
infected has a sensitive index very close to $+1$. 
This fact reinforces the importance of prevention measures, 
which protect susceptible pregnant women of becoming infected. 
 

\section{Conclusion}
\label{sec:7}

We proposed a new mathematical model for the transmission of Zika disease 
taking into account newborns with microcephaly. It has been shown that the
proposed model fits well the recent reality of Brazil from February 4, 2016 
to November 10, 2016 \cite{who:data:zika}. From a sensitivity analysis, 
we conclude that in order to reduce the number of new infection by Zika virus, 
it is important to implement control measures that reduce the average number 
of daily biting and the transmission probability from infected mosquitoes 
to susceptible pregnant women. 


\bigskip

\noindent \emph{Acknowledgements}. Area and Nieto were supported by  
\emph{Agencia Estatal de Innovaci\'on} of Spain, project MTM2016-75140-P, 
co-financed by FEDER and \emph{Xunta de Galicia}, grants GRC 2015--004 and R 2016/022. 
Silva and Torres were supported by FCT and CIDMA within project UID/MAT/04106/2013; 
by PTDC/EEI-AUT/2933/2014 (TOCCATA), co-funded by 3599-PPCDT 
and FEDER funds through COMPETE 2020, POCI. Silva is also grateful to the 
FCT post-doc fellowship SFRH/BPD/72061/2010.


\vspace*{-3mm}



\end{document}